\documentclass{article}
\usepackage{a4wide}
\usepackage{graphicx}
\usepackage{epsfig}
\usepackage{amsmath}
\usepackage{amssymb}

\newcommand{\dist}          {d}

\newcommand{\shval}          {\mathrm{shval}}

\newcommand{\frag}          {\mathrm{frag}}
\newcommand{\girth}          {\mathrm{girth}}
\newcommand{\tail}          {\mathrm{tail}}
\newcommand{\head}          {\mathrm{head}}

\newcommand{\apex}          {\mathrm{apex}}
\newcommand{\best}          {\mathrm{best}}

\DeclareMathOperator*{\polylog}{polylog}
\newtheorem{theorem}{Theorem}

\newtheorem{lemma}[theorem]{Lemma}
\newtheorem{corollary}[theorem]{Corollary}
\newtheorem{remark}[theorem]{Remark}
\newcommand{\qed}{\hfill \ensuremath{\Box}}

\newenvironment{proof}{\vspace{1ex}\noindent{\bf Proof.}\hspace{0.5em}}
	{\hfill\qed\vspace{2ex}}

\newcommand{\nota}[1]{\textbf{(*)}\marginpar {\tiny \raggedright{(*) #1}}}
\newcommand{\commento}[1] {}
\newcommand{\floor}[1]       {\left\lfloor #1 \right\rfloor}

\newcommand{\round}[1]       {\left( #1 \right)}
\newcommand{\braces}[1]       {\left\{ #1 \right\}}

\newcommand{\ResilientSpanner}{\texttt{ResilientSpanner}}
\newcommand{\ParsimoniousCycles}{\texttt{ParsimoniousCycles}}

\newenvironment{mylist}[1]{
\setbox1=\hbox{#1}
\begin{list}{}{
\setlength{\labelwidth}{\wd1}
\setlength{\leftmargin}{\wd1}
\addtolength{\leftmargin}{0em}
\addtolength{\leftmargin}{\labelsep}
\setlength{\rightmargin}{1em}}}{\end{list}}

\newcommand{\litem}[1]{\item[#1\hfill]}

\newcounter{progcount}
\newcounter{linecount}[progcount]

\newcommand{\N}{\refstepcounter{linecount}\thelinecount. \>}
\newcommand{\NL}[1]{\refstepcounter{linecount}\thelinecount. \label{#1}\>}
\newcommand{\rem}[1]{\mbox{/* \textit{#1} */}}
\newenvironment{prog}[1]{
    \refstepcounter{progcount}\label{#1}
    \par\vspace{0.5ex}\noindent\hspace{1ex}
    \begin{minipage}{\linewidth}
    \small
    \begin{tabbing}
    =spa\=spa\=spa\=spa\=spa\=spa\=spa\=spa\=spa\=spa\=spa\=spa\=\kill
}%
{
    \end{tabbing}
    \end{minipage}\\[0.5ex]
}

\newcommand{\key}[1]{\textbf{#1~}}\ignorespaces

\begin{document}
\pagestyle{plain}

\title{\bf On Resilient Graph Spanners\thanks{Work partially supported by the Italian Ministry of Education,
University, and Research (MIUR) under PRIN 2012C4E3KT national
research project ``AMANDA -- Algorithmics for MAssive and Networked
DAta''. A preliminary version of this paper was presented at the 21st Annual European Symposium on Algorithms~\cite{AFIR13}.}}

\author{Giorgio Ausiello\footnote{Dipartimento di Ingegneria Informatica, Automatica e Gestionale, Universit\`a
di Roma ``La Sapienza'', via Ariosto 25, 00185 Roma, Italy. Email: {\tt ausiello@dis.uniroma1.it}.}
\and 
Paolo G. Franciosa\footnote{Dipartimento di Scienze Statistiche, Universit\`a di Roma ``La Sapienza'',
piazzale Aldo Moro 5, 00185 Roma, Italy. Email: {\tt paolo.franciosa@uniroma1.it}.} 
\and 
Giuseppe F. Italiano\footnote{Dipartimento di Ingegneria Civile e Ingegneria Informatica,
Universit\`a di Roma ``Tor Vergata'', via del Politecnico 1, 00133
Roma, Italy. Email: {\tt giuseppe.italiano@uniroma2.it}.} 
\and 
Andrea Ribichini\footnote{Dipartimento di Ingegneria Informatica, Automatica e Gestionale, Universit\`a
di Roma ``La Sapienza'', via Ariosto 25, 00185 Roma, Italy. E-mail: {\tt ribichini@dis.uniroma1.it}.}
}

\date{}

\maketitle

\begin{abstract}
We introduce and investigate a new notion of resilience in graph spanners. 
Let $S$ be a spanner of a weighted graph $G$.
Roughly speaking, 
we say that $S$ is resilient if all its point-to-point distances are resilient to edge failures. Namely, whenever any edge in $G$ fails, then as a consequence of this failure all  distances do not degrade in $S$ substantially more than in $G$ (i.e., the relative distance
increases in $S$ are very close to those in the underlying graph $G$).
In this paper 
we show that sparse resilient spanners exist, and that they can be computed efficiently.
\end{abstract}

\section{Introduction}

Spanners are fundamental graph structures that have been extensively studied in the last decades since they were introduced in~\cite{PS89}.
Given a graph $G$, a \emph{spanner} is a (sparse) subgraph of $G$ that preserves the approximate distance between each pair of vertices.
More precisely, a $t$-spanner of a graph $G=(V,E)$ is a subgraph $S=(V,E_S)$, $E_S\subseteq E$, 
that distorts distances in $G$ up to a multiplicative factor $t$: i.e., for all vertices $x,y$,  
$d_S(x,y)\leq t\cdot d_G(x,y)$, where $d_G$ denotes the distance in graph $G$. We refer to $t$ as the \emph{stretch factor} (or \emph{distortion}) of the spanner $S$. 
It is known how to compute in $O(m+n)$ time a $(2k-1)$-spanner, with $O(n^{1+ \frac{1}{k}})$ edges~\cite{ADDJ93,HZ96} (which is conjectured to be optimal for any $k$),
where $m$ and $n$ are respectively the number of edges and vertices in the original graph $G$.
We note  that $t$-spanners are only considered for $t\geq 3$, as  $2$-spanners 
can have as many as $\Theta(n^2)$ edges.

Several other spanners have been considered in the literature.  
For $\alpha\geq 1$ and $\beta\geq 0$, an $(\alpha,\beta)$-spanner of an unweighted graph $G=(V,E)$ is a subgraph $S=(V,E_S)$, $E_S\subseteq E$, 
that distorts distances in $G$ up to a multiplicative factor $\alpha$ and an additive term $\beta$: i.e., for all vertices $x,y$,  
$d_S(x,y)\leq\alpha\cdot d_G(x,y)+\beta$.
In~\cite{Baswana2010}, it is shown how to compute a $(k,k-1)$-spanner containing $O(k \cdot n^{1+1/k})$ edges, for any integer $k \geq 2$. 
Note that $t$-spanners can be referenced to as $(t,0)$-spanners, while $(1,\beta)$-spanners are also known as \emph{purely additive spanners} ($d_S(x,y)\leq d_G(x,y)+\beta$). Algorithms for computing (1,2)-spanners with $O(n^{3/2})$ edges are given in~\cite{Aingworth,Dor,RodittyAdditive}, for (1,4)-spanners  with $\tilde{O}(n^{7/5})$ 
edges in~\cite{chechik4}, and for (1,6)-spanners with $O(n^{4/3})$ edges in~\cite{Baswana2010}.

Spanners have been investigated also in the fully dynamic setting, where edges may be added to or deleted
from the original graph. In~\cite{AFI06}, efficient dynamic deterministic algorithms are first presented for low-stretch spanners.
\commento{
\nota{PGF: attenzione, da qui in poi parliamo di termini additivi, quindi unweighted.}
a (2,1)-spanner and a (3,2)-spanner of an unweighted
graph are maintained under an intermixed sequence of $\Omega(n)$
edge insertions and deletions in $O(\Delta)$ amortized time per
operation, where $\Delta$ is the maximum vertex degree of the original
graph. The (2,1)-spanner has $O(n^{3/2})$ edges, while the
(3,2)-spanner has $O(n^{4/3})$ edges. 
}
A faster randomized dynamic
algorithm for spanners has been later proposed by
Baswana~\cite{ESA06}: given an unweighted graph, a
$(2k-1)$-spanner of expected size $O(k\cdot n^{1+1/k})$ can be maintained in
$O(\frac{m}{n^{1+1/k}} \cdot \polylog n)$ amortized expected time for
each edge insertion/deletion, where $m$ is the current number of edges in the graph.
For $k=2,3$ (i.e., 3- and
5-spanners, respectively), the amortized expected time of the randomized
algorithm becomes constant. The algorithm by Elkin~\cite{Elkin07} maintains a $(2k-1)$-spanner with expected
$O(k n^{1+1/k})$ edges in expected constant  time per edge insertion
and expected $O(\frac{m}{n^{1/k}})$ time per edge deletion. 
More recently, Baswana et al.~\cite{baswanaACM} proposed two faster fully dynamic randomized algorithms for maintaining $(2k-1)$-spanners of unweighted graphs: the expected amortized time  per insertion/deletion is $O(7^{k/2})$ for the first algorithm and $O(k^2 \log^2 n)$ for the second algorithm, and in both cases  the spanner expected size is optimal up to a polylogaritmic factor.

As observed in~\cite{STOC09}, this traditional fully dynamic model may be too pessimistic in several application scenarios,
where the possible changes to the underlying graph are rather
limited. Indeed, there are cases where there can be only temporary network failures: namely, graph edges may occasionally fail, but only for a short period of time, and it is possible to recover quickly 
from such failures. In those scenarios, rather than maintaining a fully dynamic spanner, which has to be updated after each change, one may be more interested in working with a static spanner capable of retaining many of its properties during edge deletions, i.e., capable of being resilient to transient failures.

Being inherently sparse, a spanner is not necessarily resilient to edge deletions and it may indeed lose some of its important properties during a transient failure. Indeed, let $S$ be a $t$-spanner of $G$:
if an edge $e$  fails in $G$, then 
the distortion of the spanner may substantially degrade, i.e., 
$S\setminus e$ may no longer be a $t$-spanner or even a valid spanner of $G\setminus e$, 
where $G \setminus e$ denotes the graph obtained after removing edge $e$ from $G$. 
In their pioneering work, Chechik et al.~\cite{STOC09} addressed this problem by introducing the notion of \emph{fault-tolerant spanners}, i.e., spanners that are resilient to edge (or vertex) failures.
Given an integer $f\geq 1$, a spanner is said to be $f$-edge (resp.~vertex) fault-tolerant if it preserves its original distortion under the failure of any set of at most $f$ edges (resp.~vertices). More formally, an \emph{$f$-edge (resp.~vertex) fault-tolerant $t$-spanner} of  $G=(V,E)$ is a subgraph $S=(V,E_S)$, $E_S\subseteq E$, such that for any subset $F \subseteq E$ (resp.~$F \subseteq V$), with $|F| \leq f$, and  for any pair of vertices $x,y \in V$ (resp.~$x,y \in V \setminus F$)  we have $d_{S \setminus F}(x,y) \leq t \cdot d_{G \setminus F}(x,y)$,
where $G \setminus F$ denotes the subgraph of $G$ obtained after deleting the edges (resp.~vertices) in $F$. Algorithms for computing efficiently fault-tolerant spanners can be found in~\cite{dmaa,ChechickAdditive,STOC09,Dinitz}.

The distortion is not the only property of a spanner that may degrade because of edge failures. Indeed, even when the removal of an edge cannot change the overall distortion of a spanner (such as in the case of a fault-tolerant spanner),  
it may still cause a sharp increase in some of its distances. Note that while the distortion is a \emph{global} property, distance increases are \emph{local} properties, as they are defined for pairs of vertices.
To address this problem, one would like to work with spanners that are not only \emph{globally resilient} (such as fault-tolerant spanners) but also \emph{locally resilient}. In other terms, we would like to make the 
distances between any pair of vertices in a spanner resilient to edge failures, i.e., whenever an edge fails, then the increases in distances in the spanner must be very close to the increases in distances in the underlying graph.
More formally, given a graph $G$ and an edge $e$ in $G$,
we define the \emph{fragility of edge} $e$ as the maximum relative increase in distance 
between any two vertices when $e$ is removed from $G$:
$$\frag_G(e) = \max_{x,y \in V} \left\{ \frac{\dist_{G\setminus e}(x,y)}{\dist_{G}(x,y)}\right\}$$
Our definition of fragility of an edge is somewhat reminiscent of the notion of \emph{shortcut value}, as contained in~\cite{Dagstuhl3}, where the distance increase is alternatively measured by the difference, instead of the ratio, between distances in $G\setminus e$ and in $G$. 
Note that for unweighted graphs
$\frag_G(e) \geq 2$ for any edge $e$.
The fragility of edge $e$ can be seen as a measure of how much $e$ is crucial for 
the distances in $G$, as it provides an upper bound to the increase in distance in $G$ between any pair of vertices when edge $e$ fails: the higher the fragility of $e$, the higher is the relative increase in some distance when $e$ is deleted. 

\vspace{0.3cm}
\noindent\textbf{Our contribution.}
To obtain spanners whose distances are resilient to transient edge failures, 
 the fragility of each edge in the spanner must be as close as possible to its fragility in the original graph. In this perspective, given a positive $\sigma$, we say that a spanner $S$ of $G$ is \emph{$\sigma$-resilient} if $\frag_S(e) \leq \max\{\sigma, \frag_G(e)\}$ for each edge $e \in S$, where $\sigma \geq 1$.
 Note that in case of unweighted graphs, for $\sigma=2$ this is equivalent to $\frag_S(e) = \frag_G(e)$.
 We show that finding sparse $2$-resilient spanners may be an overly ambitious goal, as we prove that there exists a family of dense graphs for which the only $2$-resilient spanner coincides with the graph itself.
It can be easily seen that, in general, spanners are not necessarily $\sigma$-resilient.
Furthermore, it can be shown that even
edge fault-tolerant multiplicative $t$-spanners are not $\sigma$-resilient, since they can only guarantee that the fragility of a spanner edge is at most $t$ times its fragility in the graph. In fact, we exhibit 1-edge fault tolerant $t$-spanners, for any $t \geq 3$, with edges whose fragility in the spanner is at least $t/2$ times their fragility in $G$.

It seems quite natural to ask whether
sparse $\sigma$-resilient spanners exist, and how efficiently they can be computed. 
We show that it is possible to compute $\sigma$-resilient $3$-spanners containing $O(W \cdot n^{3/2})$ edges for graphs with positive edge weights in $[w_{\mathrm{min}},w_{\mathrm{max}}]$, where $W = \frac{w_{\mathrm{max}}}{w_{\mathrm{min}}}$. The size is optimal for small edge weights. 
\commento{
In the case of unweighted graphs, in addition we are also able to compute 
(1,2)-spanners and (2,1)-spanners of optimal asymptotic size (i.e., containing $O(n^{3/2})$ edges)}
The total time required to compute our spanners is $O(mn + n^2 \log n)$ in the worst case. 
To compute our $\sigma$-resilient spanners, we start from a non-resilient spanner, and then add to it $O(W \cdot n^{3/2})$ edges from a carefully chosen set of short cycles in the original graph. The algorithm is simple and thus amenable to practical implementations, while the upper bound on the number of added edges is derived from sophisticated combinatorial arguments.

The same approach can be used for turning any given $t$-spanner into a $\sigma$-resilient $t$-spanner, for $\sigma \geq t > 3$.
\commento{, in case of weighted graphs, or for turning any given $(\alpha,\beta)$-spanner into a $\sigma$-resilient $(\alpha,\beta)$-spanner, for $\sigma \geq \alpha+\beta > 3$, in case of unweighted graphs. 
In all those cases, }
Once again, the total number of edges added to the initial spanners is $O(W \cdot n^{3/2})$.
Our results for $\sigma=t=3$ and for small edge weights seem to be the most significant ones, both from the theoretical and from the practical point of view. From a theoretical perspective, our $\sigma$-resilient $3$-spanners have the same asymptotic size as their non-resilient counterparts. From a practical perspective, there is empirical evidence \cite{ADFIR09} that small stretch spanners provide the best performance in terms of stretch/size trade-offs, and that spanners of larger stretch are not likely to be of practical value.
Table~\ref{tabella} puts our results in perspective with the fragility and the size of previously known spanners.

\begin{table}[h]
\begin{small}
\begin{center}
\begin{tabular}{|c|c|c|c|}
\hline 
Spanner $S$  & $\frag_{S}(e)$ & Size & Ref. \\
\hline
\hline
\begin{tabular}{cc}
$(2k-1)$-spanner, $k\geq 2$
\end{tabular}
 & unbounded & $O\left(n^{1+\frac{1}{k}}\right)$ & \cite{ADDJ93} \\
\hline
\begin{tabular}{cc}
1-edge fault-tolerant \\\
 $(2k-1)$-spanner, $k\geq 2$ 
\end{tabular}
& $\leq (2k-1) \cdot \frag_{G}(e)$ & $O\left(n^{1+\frac{1}{k}}\right)$ & \cite{STOC09}\\
\hline
\begin{tabular}{cc}
$\sigma$-resilient 
$(2k-1)$-spanner, \\\
$\sigma \geq 2k-1$, $k\geq 2$
\end{tabular}
& $\leq \max\{\sigma, \frag_{G}(e)\}$ 
&
$O\left(W \cdot n^{\frac{3}{2}}\right)$
& 
\begin{tabular}{cc}
This \\\
paper 
\end{tabular}
\\
\hline
\end{tabular}
\vspace{2mm}
\caption{Fragility and size of spanners. Factor $W$ in the last line is the ratio between maximum and minimum edge weight.}\label{tabella}
\end{center}
\end{small}
\end{table}

Also $(\alpha,\beta)$-spanners of unweighted graphs can be turned into $\sigma$-resilient $(\alpha,\beta)$-spanners, for any $\sigma \geq \alpha+\beta$, using the same technique, adding $O(n^{3/2})$ edges in the worst case.

The remainder of this paper is organized as follows. We start with few preliminary definitions and basic observations in Section~\ref{se:preliminary}. In Section~\ref{se:negative} we show some negative results on 2-resilient spanners and 1-edge fault-tolerant spanners. 
In Section~\ref{se:spanners} we describe our algorithm for computing $\sigma$-resilient spanners. In particular, we
 first describe in Section~\ref{se:trivial} a trivial approach. Next, in Section~\ref{se:size}, we show how to bound the size of  $\sigma$-resilient spanners. Finally, in Section~\ref{se:computing}, we show how to compute efficiently $\sigma$-resilient spanners.
\commento{
Our technique for computing $\sigma$-resilient spanners is then outlined in Section~\ref{se:trivial}. In Section~\ref{se:size} we prove that $\sigma$-resilient spanners containing $O(W \cdot n^{3/2})$ edges exist for $\sigma \geq 3$, and in Section~\ref{se:computing} we describe an efficient algorithm for computing them.   
}
Section~\ref{se_conclusions} lists some concluding remarks.

\section{Preliminaries}
\label{se:preliminary}

We assume that the reader is familiar with standard graph terminology.
In our paper, we deal with weighted undirected graphs, i.e., undirected graphs having weights associated to their edges. The length of a path is the sum of the weights of its edges. In unweighted graphs the length of a path is given by the number of its edges. Note that unweighted graphs can be seen as special cases of weighted graphs, where all the weights are 1.
A shortest path is a path of minimum length between two vertices, and the distance between two vertices is given by the length of a shortest path between the two vertices. Let $G=(V,E)$ be an undirected graph. Throughout this paper, we use the notation $\dist_G(u,v)$ to denote the distance between vertices $u$ and $v$ in $G$. Let $F\subseteq E$ be any subset of edges in $G$: we denote by $G\setminus F$ the graph obtained after deleting edges in $F$ from $G$. Note that, as a special case
$G\setminus e$ denotes the graph obtained after deleting edge $e$ from $G$. Similary,  we let  $G\cup H$ denote the graph obtained after adding edges in $H$ to $G$, where $H$ and $G$ have the same set of vertices. 

Let $G=(V,E)$ be an undirected 
(weighted or unweighted) graph, with $m$ edges and $n$ vertices. 
A \emph{bridge} is an edge  $e\in E$ whose deletion increases the number of connected components of $G$. 
Note that an edge is a bridge if and only if it is not contained in any cycle of $G$. Graph $G$ is \emph{2-edge-connected} if it does not have any bridges. The \emph{2-edge-connected components} of $G$ are its maximal 2-edge-connected subgraphs.
Let $e$ be an edge in $G$, and denote by  ${\cal C}_e$ 
the set of all the cycles containing $e$: if $G$ is 2-edge-connected, then ${\cal C}_e$ is non-empty for each $e\in E$.
We refer to a
shortest cycle among all cycles in ${\cal C}_e$ 
as a \emph{short cycle for} edge $e$.
Note that if $G$ is 2-edge-connected, then at least one short cycle 
always exists for any edge. Short cycles are not necessarily unique: for each $e\in E$, we denote by $\Gamma_{e}(G)$ the set of short cycles for $e$ in graph $G$.
Let $G$ be an undirected unweighted graph: the \emph{girth} of $G$, denoted by \emph{girth}$(G)$, is the length of a shortest cycle in $G$.

\commento{
\nota{PGF: queste definizioni servono solo a definire poi l'insieme di short cycles $\Gamma_e$. Potremmo definire direttamente qui $\Gamma_e$}
Similarly, we denote by 
${\cal P}_{e}(x,y)$ the set of all paths between $x$ and $y$ containing edge $e$, and by 
${\cal P}_{\overline{e}}(x,y)$ the set of all paths between $x$ and $y$ avoiding edge $e$.
We further denote by $\Pi_{e}(x,y)$ (respectively $\Pi_{\overline{e}}(x,y)$)
the set of shortest paths in ${\cal P}_{e}(x,y)$ (respectively ${\cal P}_{\overline{e}}(x,y)$).
}
The \emph{fragility} of an edge $e=(u,v)$ in graph $G$ is defined as
$\frag_G(e) = \max_{x,y \in V} \left\{\frac{\dist_{G\setminus e}(x,y)}{\dist_{G}(x,y)}\right\}$. 
Given a graph $G$ and a $t$-spanner $S$ of $G$, and given $\sigma \geq 1$, we say that edge $e$ is \emph{$\sigma$-fragile in $S$} if $\frag_S(e) > \max\braces{\sigma, \frag_G(e)}$. A $t$-spanner $R$ is \emph{$\sigma$-resilient} if $\frag_R(e) \leq \max\{\sigma, \frag_G(e)\}$ for each edge $e \in R$, i.e., if $R$ does not contain $\sigma$-fragile edges.

The following lemma shows that in the definition of fragility of an edge $e=(u,v)$, the maximum is obtained for $\{x,y\}=\{u,v\}$, i.e., exactly at its two endpoints.

\begin{lemma}\label{le:shortcutratio}
Let $G=(V,E)$ be a connected graph with positive edge weights, and let $e=(u,v)$ be any edge in $G$.
Then $\frag_G(e) = \frac{\dist_{G\setminus e}(u,v)}{\dist_{G}(u,v)}$.
\end{lemma}
\begin{proof}
Let $x$ and $y$ be any two vertices in $G$.
To prove the lemma it suffices to show that 
$\frac{\dist_{G \setminus e}(x,y)}{\dist_{G}(x,y)} \leq \frac{\dist_{G \setminus e}(u,v)}{\dist_{G}(u,v)}$.
We distinguish two cases, depending on whether there is a shortest path in $G$ between $x$ and $y$  that 
avoids edge $e$ or not.
If there is such a shortest path, then $\dist_{G \setminus e}(x,y) = \dist_{G}(x,y)$. Since $\dist_{G \setminus e}(u,v)\geq {\dist_{G}(u,v)}$,
the lemma trivially holds. 

Assume now that all shortest paths between $x$ and $y$ in $G$ go through edge $e=(u,v)$. In this case, $\dist_G(x,y) \geq \dist_G(u,v)$. 
If edge $e$ is a bridge, then  $\frag_G(e) = \frac{\dist_{G\setminus e}(u,v)}{\dist_{G}(u,v)} =+\infty$, and again the lemma trivially holds.
If $e$ is not a bridge, then the graph $G\setminus e$ is connected.
Since there is at least a (not necessarily shortest) path in $G\setminus e$ between $x$ and $y$ containing the shortest path in $G\setminus e$ from $u$ to $v$, we have that
$\dist_{G\setminus e}(x,y) \leq  \dist_G(x,y) - \dist_G(u,v) + \dist_{G\setminus e}(u,v)$, 
or equivalently
\begin{equation}\frac{\dist_{G \setminus e}(x,y)}{\dist_{G}(x,y)}
\leq
\frac{\dist_G(x, y) - \dist_G(u,v) + \dist_{G\setminus e}(u,v)}{\dist_{G}(x,y)}
\label{eq:1}
\end{equation}
Since $\dist_G(x, y) - \dist_G(u,v) + \dist_{G\setminus e}(u,v)\geq \dist_{G}(x,y)$, we can
upper bound the right-hand side of~(\ref{eq:1}) by subtracting $\dist_G(x,y) - \dist_G(u,v)\geq 0$ from both its numerator and denominator, thus yielding the lemma.
\end{proof}

Note that for unweighted graphs ${\dist_{G}(u,v)}=1$, and thus 
Lemma~\ref{le:shortcutratio} can be simply stated as $\frag_G(e) =
{\dist_{G\setminus e}(u,v)}$. 
The following simple lemma shows that, when inserting new edges into a graph $G$, the fragility of the old  edges cannot increase.

\begin{lemma}\label{le:monotone}
Let $G$ and $H$ be any pair of weighted graphs on the same set of vertices, and let  $G\cup H$ be the graph obtained after adding edges in $H$ to $G$. 
Then, $\frag_{G\cup H}(e) \leq \frag_{G}(e)$ for each edge $e$ in $G$.
\end{lemma}
\begin{proof}
Consider an edge $e=(u,v)$ in $G$. Since $G \subseteq G\cup H$,  $\dist_{G\cup H}(u,v) \leq \dist_{G}(u,v)$.
If $\dist_{G\cup H}(u,v) < \dist_{G}(u,v)$, a shortest path from $u$ to $v$ in $G\cup H$ avoids edge $e$. This is equivalent to saying that $\dist_{(G\cup H) \setminus e}(u,v) = \dist_{G\cup H}(u,v)$, and hence $\frag_{G\cup H}(e) = 1 \leq \frag_{G}(e)$.
Otherwise, $\dist_{G\cup H}(u,v) = \dist_{G}(u,v)$ and $\dist_{(G\cup H) \setminus e}(u,v) \leq \dist_{G\setminus e}(u,v)$: again, $\frag_{G\cup H}(e) \leq \frag_{G_1}(e)$. 
\end{proof}


The fragility of all edges in a graph $G=(V,E)$ with positive edge weights can be trivially computed in a total of $O(m^2 n + mn^2 \log n)$ worst-case time by simply computing, for each edge $e\in E$, all-pairs shortest paths in graph $G\setminus e$. A faster bound of $O(m n +n^2\log n)$ can be achieved by using either a careful modification of algorithm \texttt{fast-exclude} in~\cite{CamilThorup} or by applying $n$ times a modified version of Dijkstra's algorithm, as described in~\cite{Dagstuhl4}. For unweighted graphs, the above bound reduces to $O(m n)$.

\section{Some negative results}\label{se:negative}
In this section we show some negative results on 2-resilient spanners and edge fault-tolerant spanners.
We first establish that finding sparse 2-resilient spanners may be an overly ambitious goal, as there are dense graphs for which the only 2-resilient spanner is the graph itself.

\begin{theorem}\label{th:lowerbound}
There is an infinite family $\cal{F}$ of graphs such that for each graph $G \in \cal{F}$ the following properties hold:
\begin{mylist}{(1)}
	\litem{(1)}  $G$ has $\Theta(n^{\delta})$ edges, with $\delta>1.72598$, where $n$ is the number of vertices of $G$.
	\litem{(2)} No proper subgraph of $G$ is a 2-resilient spanner of $G$.
	\litem{(3)} There exists a 2-spanner $S$ of $G$ such that  $\Theta(n^{\delta})$ edges of $G\setminus S$, with $\delta>1.72598$, need to be added back to $S$ in order to make it 2-resilient.
\end{mylist}
\end{theorem}
\begin{proof}
The family $\cal{F}$ is defined as the set of graphs $\{I_3, I_6, I_9, \ldots, I_{3k}, \ldots\}$, with each $I_{3k}$ being the complement of the intersection graph of all the $k$-sets contained in a $3k$-set, $k\geq 1$. Given a set $U$, with $|U|=3k$, graph $I_{3k}$ contains a vertex $v_A$ for each subset $A \subset U$ with $|A| = k$, and vertex $v_A$ is adjacent to vertex $v_B$ if and only if $A\cap B = \emptyset$.

	Graph $I_{3k}$ has ${3k \choose k}$ vertices, each having degree ${2k \choose k}$, since this is the number of $k$-sets that can be chosen from the remaining $2k$ elements. Let $n={3k \choose k}$ be the number of vertices in $I_{3k}$: then $I_{3k}$ has
$m = \frac{n}{2} {2k \choose k}$ edges. By Stirling approximation,
we have that $m = \Theta\left(n^{1 + \frac{2}{3 \log_2 3 -2}}\right)$, where $\frac{2}{3 \log_2 3 -2} > 0.72598$. This proves Property (1).

We now turn to Property (2). We first claim that there is only one path of length 2 between any pair of adjacent vertices in $I_{3k}$. Indeed,  for any two adjacent vertices $v_A$ and $v_B$, there is exactly one vertex, namely $v_{U\setminus(A \cup B)}$, which is adjacent to both $v_A$ and $v_B$. Thus each edge belongs to exactly one triangle, which implies that the fragility of any edge in $I_{3k}$ is 2, and that there is only one path of length 2 between any pair of adjacent vertices.
Let $S\subset I_{3k}$ be a 2-spanner of $I_{3k}$. We show that $S$ is not 2-resilient.
Let $v_A$ and $v_B$ be two adjacent vertices in $I_{3k}$ such that $(v_A, v_B) \not\in S$, and let $C = U \setminus (A \cup B)$.
We know that $v_A, v_C, v_B$ is the only path of length 2 from $v_A$ to $v_B$ in $I_{3k}$. Since $S$ is a 2-spanner of $I_{3k}$, both $(v_A, v_C)$ and $(v_C, v_B)$ must be in $S$.
For the same reason above, the only path of length 2 in $I_{3k}$ from $v_A$ to $v_C$ is $v_A, v_B, v_C$,
so $\dist_{S\setminus\{(v_A, v_C)\}}(v_A, v_C) > 2$, because $(v_A, v_B) \not\in S$. Thus $\frag_S((v_A, v_C)) > 2$, while $\frag_{I_{3k}}((v_A, v_C)) = 2$, which implies that $S$ is not 2-resilient.

To prove Property (3), let $S$ be a subgraph of $I_{3k}$ obtained by deleting exactly one edge from each triangle in $I_{3k}$. Since each edge in $I_{3k}$ is contained in exactly one triangle, there is always such an $S$, and it is a 2-spanner of $I_{3k}$. Furthermore, 
by Property (1), $\frac{m}{3} =\Theta(n^{\delta})$ edges, with $\delta>1.72598$,
 have to be deleted from $I_{3k}$ in order to produce $S$. 
By Property (2),  $\Theta(n^{\delta})$ edges of $I_{3k}\setminus S$, with $\delta>1.72598$, need to be added back to $S$ in order to make it a 2-resilient 2-spanner of $I_{3k}$.
\end{proof}

Edge fault-tolerant spanners~\cite{STOC09}  provide a simple way to bound distance increases under edge faults. However, they are not $\sigma$-resilient, as the next theorem shows.

\begin{theorem}\label{th:boundedshval}
Let $G=(V,E)$ be a graph.
\begin{mylist}{(a)}
\litem{(a)}
Let $S_f$ be any 1-edge fault tolerant $t$-spanner of $G$. Then $\frag_{S_f}(e) \leq t \cdot \frag_G(e)$ for each $e \in S_f$.
\litem{(b)}
There exist 1-edge fault-tolerant $t$-spanners that are not $\sigma$-resilient, for
any $\sigma < t/2$.
\end{mylist}
\end{theorem}
\begin{proof}
We first prove (a). By definition of 1-edge fault tolerant $t$-spanner, we have  
$\dist_{S_f \setminus e}(u,v) \leq t \cdot \dist_{G \setminus e}(u,v)$ for each edge $e=(u,v)$, and since $e \in S_f$ we have 
$\dist_{S_f}(u,v) = \dist_{G}(u,v)$.
The fragility of $e$ in $S_f$ is then:
$$\frag_{S_f}(e) = \frac{\dist_{S_f \setminus e}(u,v)}{\dist_{S_f}(u,v)}
\leq
\frac{t \cdot \dist_{G \setminus e}(u,v)}{\dist_{S_f}(u,v)}
=
\frac{t \cdot \dist_{G \setminus e}(u,v)}{\dist_{G}(u,v)}
=
t \cdot \frag_G(e).$$
To prove (b), consider the graph illustrated in
Figure~\ref{fi:noresilient}. The subgraph defined by bold edges (i.e., the whole graph except edges $e_i$,
with $1 \leq i \leq t$) is a 1-edge fault tolerant $t$-spanner. The fragility of edge $e$ in the original
graph is $t$, while its fragility in the spanner is $t^2/2$, i.e., it is greater than the fragility in the original graph by a factor of $t/2$. 
\end{proof}
\begin{figure}[t]
\begin{center}
\includegraphics[width=12cm]{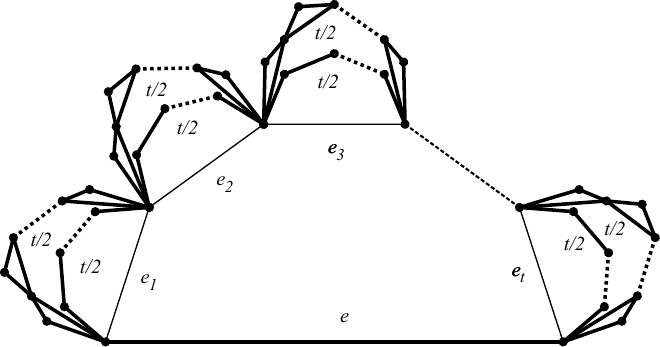}
\end{center}
\caption{A 1-edge fault tolerant $t$-spanner that is not $\sigma$-resilient for any $\sigma<t/2$. Edges $e_i$, $1 \leq i \leq t$, are not included in the $t$-spanner.}\protect\label{fi:noresilient}
\end{figure}	

\section{On $\sigma$-resilient spanners}
\label{se:spanners}

In this section we present our algorithm for computing $\sigma$-resilient spanners,
We start by describing a trivial approach in Section~\ref{se:trivial}. Then, in Section~\ref{se:size}, we introduce the notion of parsimonious sequence of cycles, which allows us to bound the size of a $\sigma$-resilient spanner. Finally, in Section~\ref{se:computing}, we show how to compute efficiently a parsimonious sequence of cycles.

\subsection{A simple-minded algorithm for computing $\sigma$-resilient spanners}\label{se:trivial}
Given a graph $G$, a $\sigma$-resilient spanner $R$ of $G$ can be computed with the following simple approach:

\begin{enumerate}
	\item Initialize $R$ to $S$, where $S$ is any $t$-spanner of $G$, with $t \leq\sigma$.
	\item For each  edge $e=(u,v)$ that is $\sigma$-fragile in $S$, 
	select a shortest path between $u$ and $v$ in $G\setminus e$ and add it to $R$.
\end{enumerate}

We refer to a shortest path between $u$ and $v$ in $G\setminus (u,v)$ as a \emph{backup path} for edge $(u,v)$. The correctness of our approach hinges on the following theorem.

\commento{
\begin{remark}\label{le:monotone}
Let $G=(V,E)$ be any weighted graph, and let  $G_1=(V,E_1)$ and $G_2=(V,E_{2})$ be any two subgraphs of $G$, with $E_{1} \subseteq E_{2} \subseteq E$. Then, $\frag_{G_2}(e) \leq \frag_{G_1}(e)$ for each edge $e \in E_1$.
\end{remark}
\begin{proof}
Consider an edge $e=(u,v) \in G_1$. Since $E_{1} \subseteq E_{2} $,  $\dist_{G_2}(u,v) \leq \dist_{G_1}(u,v)$.
If $\dist_{G_2}(u,v) < \dist_{G_1}(u,v)$, a shortest path from $u$ to $v$ in $G_2$ avoids edge $e$. This is equivalent to saying that $\dist_{G_2 \setminus e}(u,v) = \dist_{G_2}(u,v)$, and hence $\frag_{G_2}(e) = 1 \leq \frag_{G_1}(e)$.
Otherwise, $\dist_{G_2}(u,v) = \dist_{G_1}(u,v)$ and $\dist_{G_2 \setminus e}(u,v) \leq \dist_{G_1\setminus e}(u,v)$: again, $\frag_{G_2}(e) \leq \frag_{G_1}(e)$. 
\end{proof}
}

\begin{theorem}\label{th:correctweigthed}
Let $S$ be a $t$-spanner of a weighted graph $G$, and let 
$R$ be computed by adding to $S$ a backup path for each $\sigma$-fragile edge $e$ in $S$, with $\sigma \geq t$. Then $R$ is a $\sigma$-resilient $t$-spanner of $G$.
\end{theorem}
\begin{proof} 
$R$ is trivially a $t$-spanner of $G$, since it contains a
$t$-spanner $S$. It remains to show that $R$ is $\sigma$-resilient, i.e., 
$\frag_R(e) \leq \max\{\sigma, \frag_G(e)\}$ for each edge $e \in R$.
We first claim that this holds for each edge $e \in R \setminus S$.
Indeed, in this case we have that
$\dist_{R\setminus e}(u,v) \leq \dist_{S}(u,v) \leq t \cdot \dist_G(u,v)$, where the latter inequality follows immediately from the fact that $S$ is a $t$-spanner of $G$, hence also $R \setminus e$ is a $t$-spanner of $G$. This implies that
		$$\frag_R(e) = \frac{\dist_{R \setminus e}(u,v)}{\dist_R(u,v)} \leq t\cdot\frac{\dist_G(u,v)}{\dist_R(u,v)}
		\leq t \leq \sigma \leq \max\{\sigma, \frag_G(e)\}.$$

To complete the proof, it suffices to show that $\frag_R(e) \leq \max\{\sigma, \frag_G(e)\}$ for each edge $e \in S$. 
Let $e=(u,v)$ be any edge in $S$. If $\frag_S(e)\leq\sigma$, the fact that $S\subseteq R$ implies by Lemma~\ref{le:monotone} that $\frag_R(e)\leq\sigma\leq \max\{\sigma, \frag_G(e)\}$. If $\frag_S(e) > \sigma$, then 
a shortest path between $u$ and $v$ in $G\setminus e$ is added as 
a backup path for $e$, at which point the fragility of edge $e$ in the resulting graph will be equal to $\frag_G(e)$. Once again, Lemma~\ref{le:monotone} guarantees that the fragility of $e$ will never decrease after adding other backup paths for different edges. At the end, when all the backup paths have been added, we will have $\frag_R(e) \leq \frag_G(e)\leq \max\{\sigma, \frag_G(e)\}$.
\end{proof}

In the special case of unweighted graphs, Theorem~\ref{th:correctweigthed} can be extended to $(\alpha,\beta)$-spanners:
\begin{corollary}\label{coro:alfabeta}
Let $S$ be an $(\alpha, \beta)$-spanner of an unweighted graph $G$, and let 
$R$ be computed by adding to $S$ a backup path for each $\sigma$-fragile edge $e \in S$, with $\sigma \geq \alpha + \beta$. Then $R$ is a $\sigma$-resilient $(\alpha,\beta)$-spanner of $G$.
\end{corollary}
\begin{proof} 
We proceed as in the proof of  Theorem~\ref{th:correctweigthed}, except in showing 
that $\frag_R(e) \leq \max\{\sigma, \frag_G(e)\}$ for each edge $e \in R \setminus S$.
Since $R \setminus e$ is an $(\alpha, \beta)$-spanner of $G$ and $\dist_R(u,v) \geq 1$, we have

		$$\frag_R(e) = \frac{\dist_{R \setminus e}(u,v)}{\dist_R(u,v)} \leq \frac{\alpha\cdot\dist_G(u,v)+\beta}{\dist_R(u,v)}
		\leq \alpha  + \frac{\beta}{\dist_R(u,v)}\leq \alpha + \beta \leq \sigma \leq \max\{\sigma, \frag_G(e)\}.$$
\end{proof}

Note that this approach has the additional benefit of producing 
a $\sigma$-resilient spanner $R$ which inherits all monotone increasing properties of the underlying spanner $S$, i.e., all properties that are preserved under edge additions:  for example, if $S$ is fault-tolerant then $R$ is fault-tolerant too.
On the other hand, there can be several choices of backup paths for an edge with high fragility: 
if no particular care is taken in selecting suitable backup paths, we may end up with a resilient spanner of large size. In more detail, 
let $S(n)$ and $T(m,n)$  be respectively 
the number of edges of the initial spanner $S$ and
the time required for its computation, where $m$ and $n$ are respectively the number of edges and vertices in the original graph $G$.   
A trivial implementation of the above algorithm computes  a $\sigma$-resilient spanner $R$ with $O(n\cdot S(n))$ edges in a total of $O(T(m,n)+(m+n \log n)\cdot S(n))$ time.

In the next sections we will show how to refine our algorithm in order improve these bounds, by reducing both the total number of  edges added to the initial spanner $S$ and
the total time required to compute a $\sigma$-resilient spanner $R$ from $S$.

\subsection{Improving the size of $\sigma$-resilient spanners}\label{se:size}

\commento{
Since Theorem~\ref{th:correct} does not depend on how backup paths are actually chosen,
a $\sigma$-resilient $(\alpha,\beta)$-spanner $R$ can be built by adding to an initial spanner $S$ any backup path for each of its high fragility edges. We prove here that we can always obtain a $\sigma$-resilient $(\alpha, \beta)$-spanner, with $\alpha+\beta =3$, contaning
$O(n^{\frac{3}{2}})$ edges in the worst case by showing the following two properties:
}

In this section we show how to refine our algorithm in order to build a $\sigma$-resilient $3$-spanner for a graph with positive edge weights, containing
$O(W \cdot n^{\frac{3}{2}})$ edges in the worst case, where $W$ is the ratio between maximum and minimum edge weight. Our improvement is based on the following two high-level ideas:
\begin{mylist}{(1)}
\litem{(1)} Bound the number of edges with high fragility (Theorem~\ref{th:fewhighshval}).
\litem{(2)} Select carefully the shortest paths to be added as backup paths so that the total number of additional edges required is small (Theorem~\ref{th:smallcycles}).
\end{mylist}

\noindent We start by bounding the number of high fragility edges.

\begin{theorem}\label{th:fewhighshval}
Let $G=(V,E)$ be a graph with positive edge weights, an let $\sigma \geq 1$. Then, the number of edges of $G$ having fragility greater than $\sigma$ is $O(n^{1+1/\floor{(\sigma+1)/2}})$.
\end{theorem}
\begin{proof}
Let $L$ be the subgraph of $G$ containing only the edges with fragility  greater than $\sigma$ in $G$.
If $L$ contains no cycle, then $L$ has at most $(n-1)$ edges and the theorem trivially holds.

Otherwise, let $C$ be a cycle in $L$, and let $\ell$ be the number of edges in $C$.
Let $e=(u,v)$ be a maximum weight edge in $C$, and let $e_1, e_2, \ldots, e_{\ell-1}$ be the remaining edges in $C$. 
Since $L$ contains $e$ and it is a subgraph of $G$, we have by Lemma~\ref{le:monotone}:
\begin{equation}
\sigma < \frag_G(e) \leq \frag_L(e)\ .
\label{eqn:frag}
\end{equation}
We claim that it must be $\dist_{L}(u,v)= w(e)$. Indeed, if $\dist_{L}(u,v)< w(e)$, we would have $\dist_{L\setminus e}(u,v)=\dist_{L}(u,v)$, and thus by Lemma~\ref{le:shortcutratio}
$$\frag_L(e) =  \frac{\dist_{L\setminus e}(u,v)}{\dist_{L}(u,v)}=1$$
which contradicts~(\ref{eqn:frag}).
Since $\dist_{L}(u,v)= w(e)$ and $w(e_i) \leq w(e)$, for $1 \leq i < \ell-1$, we have that
$$\sigma < \frag_G(e) \leq \frag_L(e) =  \frac{\dist_{L\setminus e}(u,v)}{\dist_{L}(u,v)} \leq
\frac{\sum_{i=1}^{\ell-1} w(e_i)}{w(e)}\leq \ell -1\ .$$
This implies that any cycle in $L$ must have more than $(\sigma+1)$ edges.
Let $L'$ be the unweighted version of $L$, i.e., $L'$ has the same vertices and edges as $L$, but its edges are unweighted.
Clearly,
$$\girth(L') > \sigma + 1\ .$$
The theorem now follows directly from a result by Bondy and Simonovits~\cite{bondy}, which states that a graph with girth greater than $\sigma + 1$ contains at most $O(n^{1+1/\floor{(\sigma+1)/2}})$ edges.
\end{proof}

We now show that  the shortest paths to be added as backup paths can be suitably chosen, so that the total number of additional edges is small. In the following, we assume that our input graph $G$ is $2$-edge-connected. 
This is without loss of generality: if $G$ is not $2$-edge-connected, 
then our algorithm can be applied separately to every $2$-edge-connected component of $G$.
Let $e=(u,v)$ be an edge of high fragility in the initial $t$-spanner. Note that, in order to identify a backup path for edge $e$, we  can refer either to a shortest path between $u$ and $v$ in $G\setminus e$ or, equivalently, to a short cycle for $e$ in $G$ (i.e., the short cycle defined by one of the shortest paths in $G\setminus e$ and the edge $e$ itself). In the following, we will identify backup paths by short cycles in $G$ rather than by shortest paths in $G\setminus e$.

An ordered sequence of cycles $C_1, C_2, \ldots, C_q$ is said to be \emph{parsimonious} if the following property holds: for any pair of cycles $C_i$ and $C_j$, with $1\leq i < j \leq q$, if $C_i$ and $C_j$ have two common vertices $x$ and $y$, where $x$ and $y$ split $C_j$ into paths $P'$ and $P''$, then either 
$P' \subseteq \bigcup_{k=1}^{j-1} C_k$ or $P'' \subseteq \bigcup_{k=1}^{j-1} C_k$.
Intuitively speaking, in a parsimonious sequence of cycles, each new cycle $C_j$ reuses as much as possible portions of paths from the union of previous cycles $C_1$, $\ldots$, $C_{j-1}$.
Figure~\ref{fi:parsimonious} illustrates the notion of parsimonious sequence of cycles..
The sequence of cycles $C_1, C_2, C_3, C_4$ shown in the left side of Figure~\ref{fi:parsimonious} is not parsimonious, since each path joining $a$ and $b$ along $C_4$ is not contained into $C_1 \cup C_2 \cup C_3$.
A parsimonious sequence $C_1, C_2, C_3, C'_4$ is shown in the right side of Figure~\ref{fi:parsimonious}, where cycle $C'_4$ is shown in bold. Cycles $C'_4$ and $C_3$ intersect in two points, namely $b$ and $d$, and the path in $C'_4$ joining $b$ to $d$ through $c$ is contained in $C_1 \cup C_2 \cup C_3$. The same holds for the path joining $a$ and $c$, the two common  vertices of $C_1$ and $C'_4$, through $b$, and for the path joining $c$ and $e$, the two common vertices of $C_2$ and $C'_4$, through $d$.

\begin{figure}[t]
\begin{center}
\includegraphics[width=13cm]{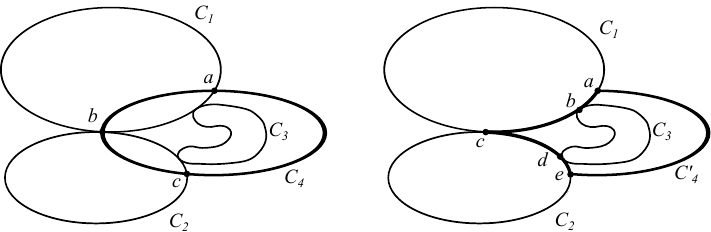}
\end{center}
\caption{The sequence of four cycles $C_1, C_2, C_3, C_4$ on the left is not parsimonious, while the sequence $C_1, C_2, C_3, C'_4$ on the right (where $C'_4$ is represented by the bold line) is parsimonious.}
\protect\label{fi:parsimonious}
\end{figure}
The following theorem bounds the total number of different edges in a parsimonious sequence of cycles.

\begin{theorem}\label{th:smallcycles}
Given a graph $G$, any parsimonious sequence $C_1, C_2, \ldots, C_q$ of cycles in $G$ contains $O(\min\{q \sqrt{n}+n, n \sqrt{q}+q\})$ edges. 
\end{theorem}
\begin{proof}
Let $V_j$ and $E_j$ be respectively the vertex set and the edge set of cycle $C_j$, for $1 \leq j \leq q$. 
We partition each edge set $E_j$ into the following three disjoint sets:
\begin{itemize}
\item \label{oldedges} $E_j^{\mathrm{old}}$ (\emph{old edges}):  
edges already in some $E_i$, $i < j$, i.e., edges in $E_j \cap \round{ \bigcup_{i =1}^{j-1} E_i }$.

\item \label{newnodes} $E_j^{\mathrm{new}}$ (\emph{new edges}): edges with at least one endpoint not contained in $\bigcup_{i =1}^{j-1} V_i$. 

\item \label{oldnodes} $E_j^{\mathrm{cross}}$ (\emph{cross edges}): edges not contained in $E_j^{\mathrm{old}}$ but with both endpoints in $\bigcup_{i =1}^{j-1} V_i$.
\end{itemize}
To prove the theorem, we have to bound the total number of edges in $\bigcup_{j=1}^{q} E_j$.
Note that we only need to bound the total number of new and cross edges (i.e., $|\bigcup_{j =1}^{q} E_j^{\mathrm{new}}|$ and $|\bigcup_{j =1}^{q} E_j^{\mathrm{cross}}|$), since each old edge in $E_j^{\mathrm{old}}$, for any $1\leq j \leq q$, is already accounted for either as a new edge or as a cross edge in some cycle $C_i$,  $i < j$.
%
Furthermore, each new edge in $E_j^{\mathrm{new}}$ can be amortized against a newly discovered vertex (i.e., a vertex $v\not\in\bigcup_{i =1}^{j-1} V_i$), and in cycle $C_j$ there can be at most two such edges which are incident to the same newly discovered vertex $v$: this implies that
$|\bigcup_{j =1}^{q} E_j^{\mathrm{new}}| \leq 2 \cdot n.$

To complete the proof, it remains to bound the total number of cross edges in the sequence.
We do this as follows. 
For each cycle $C_j$ we choose arbitrarily an orientation $\overrightarrow{C_j}$,  in one of the two possible directions, and direct its edges accordingly. Given a directed edge $e = (x,y)$, we denote its endpoint $x$ as $\tail(e)$ and its endpoint $y$ as $\head(e)$.
We build a bipartite graph $\cal{B}$ in which one vertex class represents the $n$ vertices $v_1, v_2, \ldots, v_n$ in $G$, and the other vertex class represents the $q$ directed cycles $\overrightarrow{C_1}, \overrightarrow{C_2}, ..., \overrightarrow{C_{q}}$.
There is an edge in $\cal{B}$ joining vertex $v$ and cycle $C_j$  if and only if $v$ is the tail of an edge in $E_j^{\mathrm{cross}}$ (hence, the degree of $C_j$ in $\cal{B}$ is equal to the size of $E_j^{\mathrm{cross}}$). 
Note that there is a one-to-one correspondence between edges in $\cal{B}$ and  $\bigcup_{j =1}^{q} E_j^{\mathrm{cross}}$. Thus, 
 to prove the theorem it suffices to show that the number of edges in  $\cal{B}$  is $O(\min\{q \sqrt{n}+n, n \sqrt{q}+q\})$. 

We claim that there cannot exist two vertices $x$ and $y$ that are tails of two pairs of directed edges in $E_i^{\mathrm{cross}}$ and $E_j^{\mathrm{cross}}$ (see Figure~\ref{fi:cycles}). Indeed, assuming $i < j$, the fact that the sequence of cycles is parsimonious implies that one of the two portions of $C_j$ defined by $x$ and $y$ must contain only edges in $E_j^{\mathrm{old}}$.
%
\begin{figure}[t]
\begin{center}
\includegraphics{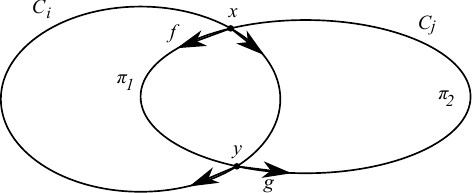}
\end{center}
\caption{On the proof of Theorem~\ref{th:smallcycles}. If $i < j$, then either edge $f$ or $g$ in $C_j$ is not in $E_j^{\mathrm{cross}}$. In fact, either $\pi_1$ or $\pi_2$ should be included in $\bigcup_{k=1}^{j-1} C_k$.}\protect\label{fi:cycles}
\end{figure}
%
The previous claim implies that the bipartite graph $\cal{B}$ does not contain  $K_{2,2}$ as a subgraph.
Determining the maximum number of edges in $\cal{B}$ is a special case of Zarankiewicz's problem~\cite{Z51}.
This problem has been solved by K\H{o}v\'ari, S\'os, Tur\'an \cite{Kovari} (see also \cite{matousek}, p.\ 65), who proved that any bipartite graph $G$ with vertex classes of size $m$ and $n$ containing no subgraph $K_{r,s}$, with the $r$ vertices in the class of size $m$ and the $s$ vertices in the class of size $n$, has
$O\left(\min\left\{m n^{1-1/r}+n, m^{1-1/s} n +m \right\}\right)$
edges, where the constant of proportionality depends on $r$ and $s$.
Since in our case the bipartite graph $\cal B$ has vertex classes of size $n$ and $q$, and $r=s=2$,  it follows that $\cal B$  contains $O(\min\{q \sqrt{n}+n, n \sqrt{q}+q\})$ edges, which yields the theorem.
\commento{
In summary, the total number of edges in the graph $\bigcup_{i =1}^q C_i$ is bounded by
$$ \left|\bigcup_{i =1}^{q}  \left(E_i^{\mathrm{old}}\cup E_i^{\mathrm{new}}\cup E_i^{\mathrm{cross}}\right)\right| \leq 2  n + O\left(\min \left\{q \sqrt{n}+n, n \sqrt{q}+q\right\}\right)$$
 and thus the theorem holds.
 }
\end{proof}

We observe that Theorem~\ref{th:smallcycles} can be adapted to provide an alternative proof of a result of Coppersmith and Elkin~\cite{Coppersmith} on distance preservers. Given a graph $G$ and $p$ pairs of vertices $\braces{(v_1,w_1), (v_2,w_2), \ldots, (v_p,w_p)}$, a pairwise distance preserver is a subgraph $S$ of $G$ such that $\dist_S(v_i,w_i) =  \dist_G(v_i,w_i)$, for $1 \leq i \leq p$. In particular, Coppersmith and Elkin~\cite{Coppersmith} showed  that it is always possible to compute a pairwise distance preserver containing $O(\min\{p \sqrt{n}+n, n \sqrt{p}+p\})$ edges. 


\subsection{Efficiently computing a parsimonious sequence of short cycles}\label{se:computing}

To compute a $\sigma$-resilient $t$-spanner of graph $G$ we start from a $t$-spanner $S$ of $G$ and add to $S$ a parsimonious sequence of short cycles, in order to apply Theorem~\ref{th:smallcycles}.
Let $E_S(\sigma)$ be the set of $\sigma$-fragile edges in $S$ (i.e., edges $e$ with $\frag_S(e) > \max\braces{\sigma, \frag_G(e)}$).
For each edge $e \in E_S(\sigma)$, we find a short cycle for edge $e$ in graph $G$ and add that cycle to $S$. 
To guarantee the parsimonious property, we 
choose in a greedy fashion short cycles that reuse paths contained in the union of previously added cycles.
We first describe how to find $\sigma$-fragile edges and next show how to compute a short cycle for each such edge.

The computation of $\sigma$-fragile edges can be accomplished in $O(m n + n^2 \log n)$ worst-case time by using an algorithm by Brandes for computing shortcut values (described in~\cite{Dagstuhl4}, Section 4.2.2). 
We recall here that, given an edge $e$ in graph $G$,  
 the \emph{shortcut value} of $e$, denoted by $\shval_G(e)$, is defined as the maximum distance increase between any two vertices after the deletion of $e$:
$$\shval_G(e) = \max_{x,y \in V}\braces{\dist_{G\setminus e}(x,y) - \dist_{G}(x,y)}\ .$$
Brandes' algorithm performs $n$ Dijkstra-like visits,  each time considering a different vertex as root; when starting a new visit from a root $r$, the algorithm computes the shortcut value of each edge incident to $r$.
It can be seen (\cite{Dagstuhl3}, Section 3.6.3) that the distance increase after the deletion of edge $e=(u,v)$ is maximum for the endpoints of $e$, i.e.,
$$\shval_G(e) = \dist_{G\setminus e}(u,v) - \dist_{G}(u,v)\ .$$
As a consequence of Lemma~\ref{le:shortcutratio}, the fragility of an edge $e=(u,v)$ can be expressed as
$$\frag_{G}(e) = \frac{\shval_G(e)}{\dist_G(u,v)} + 1\ ,$$
and thus it can be easily determined once the shortcut value of the same edge $e$ is known.

A trivial algorithm for computing a parsimonious sequence of short cycles can be obtained as follows. For each edge $(u,v) \in E_{S}(\sigma)$, we compute a shortest path from $u$ to $v$ in $G\setminus (u,v)$: when comparing paths with the same weight, we select the path containing the smallest number of new edges (ties can be broken arbitrarily). This can be done by means of a slight modification of Dijkstra's algorithm and it requires a total of $O\round{|E_{S}(\sigma)| \cdot (m + n \log n)}$ worst-case time, which is  
$O\round{n^{1+1/\floor{(\sigma+1)/2}} \cdot (m + n \log n)}$ by Theorem~\ref{th:fewhighshval}.

We next show how to improve this time to $O\round{m n + n^2 \log n}$ in the worst case.
Consider Algorithm \ResilientSpanner\ illustrated in Figure~\ref{fi:algobackup}. For each vertex $r$ in $G$, we first compute all the $\sigma$-fragile edges incident to $r$. Next, we augment the current spanner with 
one short cycle for each $\sigma$-fragile edge. Throughout, we will guarantee the important property that the total sequence of short cycles added during all the iterations of Algorithm \ResilientSpanner\ is parsimonious. 
This will be accomplished by  Algorithm \ParsimoniousCycles, invoked in Line~\ref{row:add}, which is the crux of the method. 
\commento{
\begin{figure}[ht]
\noindent\hrulefill%
\begin{prog}{pr:resilientspanner}
\key{Algorithm} ResilientSpanner($G$, $S$, $\sigma$, $E_S(\sigma)$)\\
\key{input:}\\
\>  graph $G$\\
\>  a $t$-spanner $S$ of $G$\\
\>  a fragility threshold $\sigma$, with $\sigma \geq t$\\
\>  the set of $\sigma$-fragile edges $E_S(\sigma)$\\
\key{output:}\\
\>  a $\sigma$-resilient $t$-spanner $R$ of $G$, with $R \supseteq S$ \\[1mm]

\vspace{0.3cm}\\
\N  \key{let}$R = S$\\
\N  \key{for each}vertex $r$ in $G$\\
\N  \>  \key{let}$E_r$ be the set of $\sigma$-fragile edge in $R$ incident to $r$\\
\NL{row:cycle}  \>  \key{for each} edge $e=(r,x) \in E_r$,\\
\NL{row:following}  \>  \>  add to $R$ a short cycle $c_x \in \Gamma_{e}(G)$ such that\\
\>  \>  \>  the sequence of added cycles is parsimonious\\
\end{prog}

\caption{Algorithm ResilientSpanner.}
\protect\label{fi:algobackup}
\noindent\hrulefill%
\end{figure}
}

\begin{figure}[ht]
\noindent\hrulefill%
\begin{prog}{pr:resilientspanner}
\key{Algorithm} ResilientSpanner($G$, $S$, $\sigma$)\\
\key{input:}\\
\>  graph $G$\\
\>  a $t$-spanner $S$ of $G$\\
\>  a fragility threshold $\sigma$, with $\sigma \geq t$\\
\key{output:}\\
\>  a $\sigma$-resilient $t$-spanner $R$ of $G$, with $R \supseteq S$ \\[1mm]

\vspace{0.3cm}\\
\N  \key{let}$S' = S$\\
\NL{row:loop}  \key{for each}vertex $r$ in $G$\\
\N  \>  \key{let}$E_r$ be the set of $\sigma$-fragile edges in $S'$ incident to $r$\\
\NL{row:add}    \>  $S' =  S' \cup \mathrm{ParsimoniousCycles}(G, S', r, E_r)$\\
\N  \key{return}$R = S'$\\
\end{prog}

\caption{Algorithm \ResilientSpanner.}
\protect\label{fi:algobackup}
\noindent\hrulefill%
\end{figure}

Before describing the Algorithm \ParsimoniousCycles\ in detail, we need few preliminary definitions. 
Let $S'$ be the spanner at a generic iteration of the loop of Algorithm \ResilientSpanner. Given a vertex $r$, we define a shortest path tree $T_r$ of $G$ rooted at $r$ to be \emph{parsimonious} if the following condition holds:
\begin{quote}
for each vertex $x\not= r$, the path from $r$ to $x$ in $T_r$ is a shortest path 
between $r$ and $x$ in $G$ with the smallest number of edges in $G\setminus S'$.
\end{quote}
It can be seen that, if $T_r$ is parsimonious, then for each pair of vertices $x',x''$ such that $x'$ is an ancestor of $x''$ in $T_r$, also the path from $x'$ to $x''$ in $T_r$ is a shortest path in $G$ with the smallest number of edges in $G\setminus S'$.
Note that a parsimonious shortest path tree $T_r$ can be computed using a slight modification of the shortest path algorithm by Dijkstra: whenever two or more alternative paths with the same weight reach a vertex, we select the path with the smaller number of edges in $G \setminus S'$ (ties can be broken arbitrarily).
\commento{
For each vertex $x \not= r$ we define:
\begin{itemize}
\item $\delta_r(x)$: the distance from $r$ to $x$;
\item $\tau_r(x)$: the vertex such that $(r, \tau_r(x))$ is the first edge in the shortest path from $r$ to $x$ in $T_r$;
\item $p_{\delta,r}(x)$: the vertex preceding $x$ in the shortest path from $r$ to $x$ in $T_r$;
\item $k_{\delta,r}(x)$: the number of edges from $G\setminus R$ in the shortest path in $T_r$ from $r$ to $x$.
\end{itemize} 
Moreover, given a vertex $v$ adjacent to $r$, let $t_r(v)$ be the subtree of $T_r$ rooted at $v$.
Let $x$ be a vertex in $t_r(v)$, and let $\pi_x$ be a shortest path from $r$ to $x$ in $G\setminus (r,v)$.
We say that $\pi_x$ is a \emph{best backup path} for $x$ if it contains the smallest number of edges from $G \setminus R$.
}
Moreover, let $v$ be a vertex adjacent to $r$ in $T_r$ and let  $\Pi_r(v)$ be the set of all shortest paths from $r$ to $v$ in $G\setminus (r,v)$.
We denote a 
path $\pi_v \in \Pi_r(v)$ having the smallest number of edges from $G \setminus S'$ as a \emph{best backup path} for edge $(r,v)$.
By definition, a best backup path for $(r,v)$ does not contain edge $(r,v)$ and thus it must contain at least one edge in $G\setminus T_r$. 
The following lemma shows that there must exist a best backup path for $(r,v)$ containing exactly one edge of $G\setminus T_r$. 

%

\begin{lemma}\label{le:bestbackup}
Let $v$ be a vertex adjacent to $r$ in $T_r$. There is at least one best backup path for $(r,v)$ that contains exactly one edge of $G\setminus T_r$.
\end{lemma}
\begin{proof}
Let $\pi_v$ be a best backup path for $(r,v)$, and let $T_r(v)$ be the subtree of $T_r$ rooted at $v$.
\begin{figure}[t]
\begin{center}
\includegraphics[width=9cm]{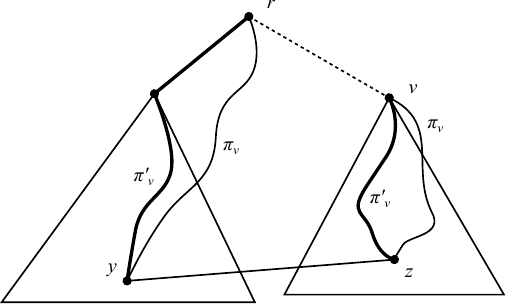}
\end{center}
\caption{On the proof of Lemma~\ref{le:bestbackup}. Solid paths are contained in $T_r$.}\protect\label{fi:bestbackup}
\end{figure}
Walk along the path $\pi_v$ starting from the root $r$, and let $z$ be the first vertex encountered in the subtree $T_r(v)$
 (see Figure~\ref{fi:bestbackup}). Let  $y$ be the vertex immediately before $z$ in $\pi_v$; clearly $y$ is not in $T_r(v)$ and  $(y,z)$ is an edge of $G \setminus T_r$.

Let $\pi(r,y)$ be the path between $r$ and $y$ in $T_r$.
Since $T_r$ is a parsimonious shortest path tree, $\pi(r,y)$ 
is a shortest path between $r$ and $y$ in $G$ with the smallest number of edges in $G\setminus S'$.
The same argument holds for the path $\pi(z,v)$ between $z$ and $v$ in $T_r$.
Hence, the path $\pi'_v=\pi(r,y)\cdot (y,z)\cdot \pi(z,v)$ obtained by concatenating the path between $r$ and $y$ in $T_r$, edge $(y,z)$, and the path between $z$ and $v$ in $T_r$ 
must be a best backup path for $(r,v)$, and it contains only one edge, i.e., $(y,z)$, from $G \setminus T_r$.
\end{proof}

We denote by \emph{single-cross backup paths} the best backup paths
that contain exactly one edge from $G\setminus T_r$ (as in Lemma~\ref{le:bestbackup}). Single-cross backup paths will be a crucial ingredient for finding efficiently a parsimonious sequence of short cycles.
\commento{For each vertex $x$ we define:
\begin{itemize}
\item $\vartheta_r(x)$: the length of the current best backup path from $r$ to $x$;
\item $p_{\vartheta,r}(x)$: the vertex preceding $x$ in the current best backup path from $r$ to $x$;
\item $k_{\vartheta,r}(x)$: the number of edges from $G\setminus R$ in the current best backup path from $r$ to $x$.
\end{itemize}
Values of $\vartheta_r(x)$ and $k_{\vartheta,r}(x)$ are computed by Algorithm~\ref{pr:backuppaths}.
}
Note that, given a root $r$ and a vertex $v$ adjacent to $r$ in $T_r$, combining the edge $(r,v)$ with a single-cross backup path for $(r,v)$ yields a short cycle for edge $(r,v)$.

While computing a parsimonious shortest path tree $T_r$ in $O(m+n\log n)$ time, in the same bound we can compute and store in each vertex $x$ the following information:
\begin{itemize}
\item $\delta(x)$: the distance from $r$ to $x$;
\item $k(x)$: the number of edges from $G\setminus S'$ in the (shortest) path from $r$ to $x$ in $T_r$;
\item $\apex(x)$: the vertex such that $(r, \apex(x))$ is the first edge in the (shortest) path from $r$ to $x$ in $T_r$;
\item $p(x)$: the vertex immediately before $x$ in the (shortest) path from $r$ to $x$ in $T_r$.
\end{itemize} 

We are now ready to complete the low-level details of Algorithm \ResilientSpanner\ of Figure~\ref{fi:algobackup}
by showing how to implement 
Algorithm  \ParsimoniousCycles, whose pseudo-code is illustrated in Figure~\ref{fi:mainalgo}. We first sketch the main ideas behind the algorithm. We are given the current spanner $S'$, a vertex $r$, and the set $E_r$ of $\sigma$-fragile edges incident to $r$. 
In the following, we denote  
by $\pi_v$ a single-cross backup path for $(r,v)$.
The objective of Algorithm  \ParsimoniousCycles\ is to  compute $\pi_v$ for each vertex $v$ such that $(r,v)\in E_r$.
By Lemma~\ref{le:bestbackup}, $\pi_v$ must be of the form
$\pi_v=\pi(r,y)\cdot (y,z)\cdot\pi(z,v)$, where $\pi(u,w)$ denotes the unique path between vertices $u$ and $w$ in $T_r$, and 
$(y,z)$ is an edge in $G\setminus T_r$ with  $\apex(y) \not= v$ and $\apex(z) = v$.
To compute $\pi_v$, 
 it thus suffices to identify such an edge $(y,z)$ in $G\setminus T_r$,  so that the 
following two properties hold:
\begin{mylist}{(1) }
\litem{(1)} The path $\pi_v=\pi(r,y)\cdot (y,z)\cdot\pi(z,v)$ is a shortest path in $G\setminus (r,v)$. Note that the weight of the path $\pi(r,y)\cdot (y,z)\cdot\pi(z,v)$ can be  computed in constant time as $(\delta(y)+w(y,z)+(\delta(z)-w(r,v)))$.
\litem{(2)} Among all shortest paths between $r$ and $v$ in $G\setminus (r,v)$, $\pi_v$ has the smallest number of edges from $G \setminus S'$. Note that the number of edges from $G \setminus S'$ in $\pi(r,y)\cdot (y,z)\cdot\pi(z,v)$ can be computed in constant time as $(k(y)+k(z))$ if $(y,z)\in S'$, and as $(k(y)+k(z)+1)$ otherwise.
\end{mylist}

Having this in mind, Algorithm \ParsimoniousCycles\ works
as follows. It stores in 
 $\best(v)$ the currently best single-cross backup path computed for each edge $(r,v)\in E_r$, where $\best(v)$ is initialized in Lines 2--3 of Figure~\ref{fi:mainalgo}.
 Next, the algorithm
scans all edges $(y,z)$ in $G\setminus T_r$
such that $(r,\apex(z))\in E_r$, as illustrated in Lines 4--5. Note that $\apex(y)\neq \apex(z)$ is further checked on Line 5, since by Lemma~\ref{le:bestbackup} when $\apex(y)=\apex(z)$ then the edge $(y,z)$ cannot be in a single-cross backup path. Otherwise, the edge $(y,z)$ can potentially induce a single-cross backup path $\gamma=\pi(r,y)\cdot (y,z)\cdot\pi(z,\apex(z))$ for edge $(r,\apex(z))$: if $\gamma$ improves the previously known value for $\best(\apex(z))$, then $\best(\apex(z))$ gets updated in Line 8. 
At the end of the loop in Lines 4--8, the algorithm has computed $\pi_v=\best(v)$ for all edges $(r,v)\in E_r$. On Lines 9--11 it returns the corresponding short cycles.

\commento{
Due to Lemma~\ref{le:bestbackup}, we can find a best backup path for edge $(r,v) \in E_r$, by looking at the ``best possible'' path composed by a path in $T_r$ from $v$ to some vertex $z$, a path in $T_r$ from $r$ to a vertex $y$, and edge $(y,z)$, with $\apex(z) = v$ and $\apex(y) \not= v$, where ``best possible'' means, among all paths having minimum length, the one that contains the minimum number of edges in $G \setminus S'$. Given any two vertices $x$ and $y$, we denote by $\pi(x,y)$ the unique path in $T_r$ between $x$ and $y$.
}

\begin{figure}[ht]
\noindent\hrulefill%
\begin{prog}{pr:backuppaths}
\key{Algorithm} ParsimoniousCycles($G$, $S'$, $r$, $E_r$)\\
\key{input:}\\
\>  graph $G$,\\
\>  $t$-spanner $S'$ of $G$,\\
\>  the root vertex $r$,\\
\>  the set $E_r$ of $\sigma$-fragile edges in $S$ incident to $r$\\

\key{output:}\\
\>  a set of short cycles, one for each edge in $E_r$\\
\\[1mm]
\NL{li:tree}  compute a parsimonious shortest path tree $T_r$\\
\NL{li:firstloop}  \key{for each}$v$ such that $(r,v) \in E_r$\\
\N  \>  \key{set}the current backup path $\best(v) = \emptyset$\\

\NL{li:nodeloop}  \key{for each}$y \in V \setminus r$\\
\NL{li:edgeloop}  \>  \key{for each}edge $(y,z) \in G \setminus T_r$ with $\apex(y) \not= \apex(z)$ and such that $(r,\apex(z))\in E_r$\\
\N  \>  \>\key{let}$\gamma = \pi(r,y) \cdot (y,z) \cdot \pi(z,\apex(z))$ \rem{where $\cdot$ denotes path concatenation}\\
\NL{li:isbetterpath}  \>  \>  \key{if}path $\gamma$ improves over $\best(\apex(z))$\\
\N  \>  \>  \> $\best(\apex(z)) = \gamma$\\
\N \key{let}$C_r = \emptyset$\\
\NL{li:lastloop}  \key{for each}$v$ such that $(r,v) \in E_r$\\
\N  \>  $C_r = C_r \cup \braces{\best(v) \cup (r,v)}$\\
\N \key{return}$C_r$\\
\end{prog}

\caption{Algorithm \ParsimoniousCycles.}
\protect\label{fi:mainalgo}
\noindent\hrulefill%
\end{figure}

\commento{
\begin{figure}[ht]
\noindent\hrulefill%
\begin{prog}{pr:backuppaths}
\key{Algorithm} FindBackupPaths($G$, $R$, $r$, $F_r$, $T_r$)\\
\key{input:}\\
\>  graph $G$,\\
\>  $(\alpha, \beta)$-spanner $R$ of $G$,\\
\>  the root vertex $r$,\\
\>  the set $F_r$ of vertices adjacent to $r$ computed by Algorithm~\ref{fi:algobackup},\\
\>  the shortest pathtree $T_r$ fulfilling Property~\ref{pr:fewnewedges}\\
\key{output:}\\
\>  for each vertex $v \in F_r$, sets parents for each vertex in a best backup path for $v$\\
\\[1mm]
\N  initialize the priority queue: $Q = \emptyset$\\
\N  \key{for each}$v \not= r \in V$\\
\N  \>  \key{set}the current backup path $\Pi(v) = \emptyset$\\
\N  \>  \key{insert}$v$ into $Q$ with priority $\delta_r(v)$\\

\N  \key{while}$Q \not= \emptyset$\\
\N  \>  \key{pop}the best priority vertex $z$ from $Q$\\
\N  \>  \key{for each}edge $(z,y) \in G$ with $\tau_r(y) \not= \tau_r(z)$\\
\NL{li:isbetterpath}  \>  \>  \key{if}path $\pi_{T_r}(r,y,z,v)$ improves over $\Pi(v)$\\
\N  \>  \>  \> \key{set}$\Pi(v) = \pi_{T_r}(r,y,z,v)$\\
\end{prog}

\caption{Algorithm FindBackupPaths.}
\protect\label{fi:mainalgo}
\noindent\hrulefill%
\end{figure}

}

\commento{
The length and the number of edges in $G \setminus S'$ of path $\gamma$ to be considered at Line~\ref{li:isbetterpath} can be easily computed in constant time using values $\delta(y)$, $k(y)$, $\delta(z)$, $k(z)$ together with the weights of edges $(y,z)$ and $(r, \apex(z))$, knowing whether $(y,z) \in S'$ and/or $(r, \apex(z)) \in S'$.
}

The next theorem shows that the set of short cycles computed in Algorithm \ResilientSpanner\ by the $n$ calls to 
\ParsimoniousCycles\ yields a parsimonious sequence (of short cycles).

\begin{theorem}\label{th:isparsimonious}
There exists an ordering of the short cycles computed by Algorithm \ResilientSpanner\ so that the resulting sequence is parsimonious.
\end{theorem}
\begin{proof}
Let $r_1, r_2, \ldots, r_n$ be the order in which the vertices in $G$ are considered as roots by Algorithm \ResilientSpanner\ (on Line 2 in Figure~\ref{fi:algobackup}). For each root $r_i$, with $1 \leq i \leq n$, 
let $E_{r_i}$ be the set of $\sigma$-fragile edges incident to $r_i$ in the original spanner $S$. Note that Algorithm
\ResilientSpanner\ computes
 a short cycle $C_{i,j}$  
 for each edge $(r_i,v_{i,j}) \in E_{r_i}$, with $1 \leq j \leq |E_{r_i}|$. Moreover, the short cycle 
 $C_{i,j}$
consists of a single-cross backup path for edge $(r,v_{i,j})$ and the edge $(r_i,v_{i,j})$ itself. To prove the theorem, we show that the sequence of short cycles, $C_{i,j}$, for $1 \leq i \leq n$ and $1 \leq j \leq |E_{r_i}|$, sorted lexicographically by increasing $(i,j)$, is parsimonious.

Let $C_{i,j}$ and $C_{i',j'}$ be any two short cycles in this sequence that share two vertices, with the pair $(i',j')$ preceding pair $(i,j)$ lexicographically. We now distinguish two cases, depending on whether  ($i' < i$) or ($i' = i$ and $j'<j$).

\begin{description}
\item{Case $i' < i$:} when cycle $C_{i,j}$ is computed by Algorithm \ParsimoniousCycles, all the edges in $C_{i',j'}$ are already in the current spanner $S'$.
Recall that $C_{i,j}$ is a short cycle for edge $(r_i,v_{i,j})$, and let $(y,z)$ be the unique edge in $C_{i,j}$ that does not belong to $T_{r_i}$. By Lemma~\ref{le:bestbackup},
the short cycle $C_{i,j}$ consists of the edge $(y,z)$ plus two subpaths in $T_{r_i}$: a path $\pi_1$ from $r_i$ to $y$ and a path $\pi_2$ consisting of  edge  $(r_i,v_{i,j})$ followed by a path from $v_{i,j}$ to $z$ (see Figure~\ref{fi:parsimoniousDiffRoot}).  
\begin{figure}[t]
\begin{center}
\includegraphics[width=11cm]{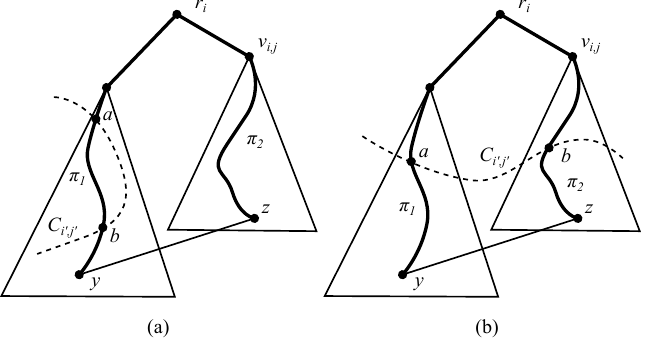}
\end{center}
\caption{Proof of Theorem~\ref{th:isparsimonious}. Case $i' < i$.}\protect\label{fi:parsimoniousDiffRoot}
\end{figure}
If $C_{i,j}$ and $C_{i',j'}$ share two vertices $a$ and $b$, two cases may occur: either $a$ and $b$ are in the same subpath
or they are in two different subpaths of $C_{i,j}$. Furthermore, 
since both $C_{i,j}$ and $C_{i',j'}$ are short cycles, the portion of $C_{i,j}$ from $a$ to $b$ must have the same weight as the portion of $C_{i',j'}$ from $a$ to $b$.
We now consider the two cases separately.


In the first case, assume without loss of generality that both $a$ and $b$ are in $\pi_1$ (see Figure~\ref{fi:parsimoniousDiffRoot} (a)). 
Note that at this point  the portion of the cycle $C_{i',j'}$ between $a$ and $b$ is already contained in $S'$.
Since $T_{r_i}$ is a parsimonious shortest path tree, the subpath $\pi_1$ in $T_{r_i}$ must contain the minimum number of edges in $G \setminus S'$, which implies that 
also the portion of the subpath $\pi_1$ between $a$ and $b$ must be contained in $S'$ (i.e., cycles $C_{i',j'}$ and $C_{i,j}$ satisfy the condition for being part of a parsimonious sequence).

In the second case, assume without loss of generality that $a$ is in $\pi_1$ and $b$ is in $\pi_2$ (see  Figure~\ref{fi:parsimoniousDiffRoot} (b)). Once again, at this point  the portion of the cycle $C_{i',j'}$ between $a$ and $b$ is already contained in $S'$. When $C_{i,j}$ is computed by Algorithm \ParsimoniousCycles, the portion of $C_{i,j}$ between $a$ and $b$ passing through edge $(y,z)$ must be contained in $S'$ (otherwise the portion of $C_{i',j'}$ between $a$ to $b$ would have produced a cycle with the same weight and fewer edges of  $G\setminus S'$).
Once again, cycles $C_{i',j'}$ and $C_{i,j}$ satisfy the condition for being part of a parsimonious sequence.

\commento{An ordered sequence of cycles $C_1, C_2, \ldots, C_q$ is said to be \emph{parsimonious} if the following property holds: for any pair of cycles $C_i$ and $C_j$, with $1\leq i < j \leq q$, if $C_i$ and $C_j$ have two common vertices $x$ and $y$, where $x$ and $y$ split $C_j$ into paths $P'$ and $P''$, then either 
$P' \subseteq \bigcup_{k=1}^{j-1} C_k$ or $P'' \subseteq \bigcup_{k=1}^{j-1} C_k$.
Intuitively speaking, in a parsimonious sequence of cycles, each new cycle $C_j$ reuses as much as possible portions of paths from the union of previous cycles $C_1$, $\ldots$, $C_{j-1}$.
}

\item{Case $i' = i$ and $j'<j$:} cycles $C_{i,j}$ and $C_{i,j'}$ are computed during the same call of Algorithm \ParsimoniousCycles\ from a root $r_i$. Since Algorithm \ParsimoniousCycles\ computes single-cross backup paths, each short cycle produced passes through the root $r_i$ and traverses exactly two subtrees of $r_i$. 
We observe that if $C_{i,j}$ and $C_{i,j'}$ do not share any subtree, then they can intersect only at the root $r_i$.
Hence, only the following two cases are possible for short cycles $C_{i,j}$ and $C_{i',j'}$ intersecting at two vertices:
\begin{itemize}
\item{$C_{i,j}$ and $C_{i,j'}$ are contained in the same two subtrees:} this case is shown in Figure~\ref{fi:parsimoniousSameRoot} (a). The intersection among cycles $C_{i,j}$ and $C_{i,j'}$ is the path in $T_{r_i}$ joining the lowest common ancestor $a$ of $y$ and $z'$ and the lowest common ancestor $b$ of $y'$ and $z$, where $(y,z)$ (resp., $(y',z')$) is the cross edge for $C_{i,j}$ (resp., $C_{i,j'}$). In this case, any relative ordering among $C_{i,j}$ and $C_{i,j'}$ produces a parsimonious sequence. 

\item{$C_{i,j}$ and $C_{i,j'}$ share one subtree:} this case is shown in Figure~\ref{fi:parsimoniousSameRoot} (b). Also in this case, cycles $C_{i,j}$ and $C_{i,j'}$ share exactly a path, namely the path between the root $r_i$ and the lowest common ancestor of $z$ and $y'$. The same argument as in the previous case applies. 
\end{itemize}

\end{description}

\begin{figure}[t]
\begin{center}
\includegraphics[width=11cm]{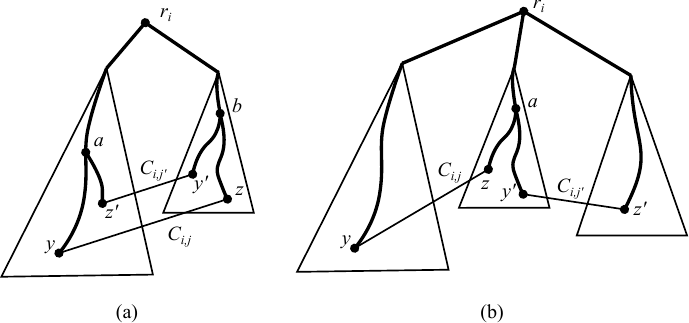}
\end{center}
\caption{Proof of Theorem~\ref{th:isparsimonious}. Case $i' = i$.}\protect\label{fi:parsimoniousSameRoot}
\end{figure}

\end{proof}

The following theorems bound the running time of Algorithm \ResilientSpanner\ and the number of edges in the computed $\sigma$-resilient $t$-spanner.

\begin{theorem}\label{th:shvalpreservingbis}
Algorithm \ResilientSpanner\ runs in $O(mn + n^2 \log n)$ time in the worst case.
\end{theorem}
\begin{proof}
We first bound the time required by Algorithm \ParsimoniousCycles\ in Figure~\ref{fi:algobackup}. 
The parsimonious shortest path tree $T_r$ in Line~\ref{li:tree} can be computed by a slight modification of Dijkstra's shortest path algorithm in $O(m + n \log n)$ worst case time, together with the auxiliary information about 
 $\delta(x)$, $k(x)$, $\apex(x)$ and $p(x)$ for each vertex $x\in T_r$. Using this auxiliary information, each edge $(y,z)$ can be processed in constant time in Lines 5--8. This implies that the total time spent through the loop in Lines 4--8 is bounded by $O(m)$.
Also the time spent in the initialization (Lines 2--3) and for returning all short cycles (Lines 9--12) is $O(m)$. As a result, each call to 
Algorithm \ParsimoniousCycles\ can be implemented in time $O(m+n\log n)$ in the worst case.

We now turn to Algorithm \ResilientSpanner\ of Figure~\ref{fi:mainalgo}. As it was previously mentioned, 
all the $\sigma$-fragile edges (Lines 2--3) can be computed in $O(m n + n^2 \log n)$ worst-case time by using Brandes' algorithm for computing shortcut values~\cite{Dagstuhl4}. 
Since each call to Algorithm \ParsimoniousCycles\  requires $O(m + n \log n)$ worst-case time, the overall running time of the algorithm is $O(mn + n^2 \log n)$ in the worst case.
\end{proof}


\begin{theorem}\label{th:shvalpreserving}
Let $G$ be a graph with $n$ vertices, with positive edge weights in $[w_\mathrm{min}, w_\mathrm{max}]$, and let $W = \frac{w_\mathrm{max}}{w_\mathrm{min}}$.
Algorithm \ResilientSpanner\ computes a $\sigma$-resilient $t$-spanner $R$ of $G$, $\sigma \geq t$, with $R\supseteq S$, containing $O\left( W \cdot n^{3/2}\right)$ edges.
\end{theorem}
\begin{proof}
By Theorem~\ref{th:correctweigthed}, the subgraph $R$ computed by Algorithm \ResilientSpanner\ is a $\sigma$-resilient $t$-spanner of $G$, since it is obtained by adding a parsimonious sequence $\cal{C}$ of short cycles to a $t$-spanner $S$, one for each $\sigma$-fragile edge $e$ in $S$.

Let $C_e$ be the cycle in $\Gamma_e(G)$ added to the spanner, and let $S(n)$ be the number of edges in $S$.
We partition $\sigma$-fragile edges $e \in S$ into three subsets, $E_{\ell}$, $E_{m}$ and $E_{h}$, according to their fragility in $G$. For each subset we separately bound the number of edges in the union of cycles in $\cal{C}$.
\begin{description}
\item[low fragility edges:] $E_{\ell} = \left\{ e \in S \ \ | \ \  \sigma < \frag_G(e) \leq 5 \right\}$. By Theorem~\ref{th:fewhighshval}, we have
$$|E_{\ell}| = O\left(\min\braces{S(n), n^{1+\frac{1}{\floor{\frac{\sigma+1}{2}}}}}\right) 
= O\left(n^{1+\frac{1}{\floor{\frac{\sigma+1}{2}}}}\right) \ ,$$
Thus, if $3 \leq \sigma < 5$ we have
$|E_{\ell}| = O\left(\min\braces{S(n), n^{3/2}}\right)$, while $E_\ell = \emptyset$ for $\sigma > 5$.
Let $e$ be any edge in $E_{\ell}$. Since $\frag_G(e) \leq 5$, cycle $C_e$ contains at most $5 W+1$ edges ($C_e$ contains exactly $\frag_G(e) \cdot W+1$ edges when $w(e) = w_\mathrm{max}$ and all other edges in $C_e$ have weight $w_\mathrm{min}$). So, we have
$$\left|\bigcup_{e \in E_\ell} C_e \right| = O\left(W \cdot n^{3/2}\right)\ \ 
\mathrm{for}\ \ 3 \leq \sigma < 5\ ,$$
while $\left|\bigcup_{e \in E_\ell} C_e \right| = 0$ for $\sigma\geq 5$;

\item [medium fragility edges:] $E_m = \left\{ e \in S \ \ | \ \  5 < \frag_G(e) < \log n \right\}$. By Theorem~\ref{th:fewhighshval}, since the fragility of each edge in $E_{m}$ is greater than $\max\braces{\sigma, 5}$, then
$$|E_{m}| = O\round{\min\braces{S(n),n^{1 + \frac{1}{\floor{(\sigma + 1)/2}}}, n^{\frac 4 3}}}
= O\round{\min\braces{n^{1 + \frac{1}{\floor{(\sigma + 1)/2}}}, n^{\frac 4 3}}}\ .$$
Each cycle $C_e$, with $e \in E_{m}$, contains at most $\log n \cdot W+1$ edges, so we have

$$\left|\bigcup_{e \in E_m} C_e \right| = O\round{W \cdot \log n \cdot \min\braces{n^{1 + \frac{1}{\floor{(\sigma + 1)/2}}}, n^{\frac 4 3}}}$$

\item [high fragility edges:] $E_h = \left\{ e \in S \ \ | \ \  \frag_G(e) \geq \log n\right\}$. By Theorem~\ref{th:fewhighshval}, $|E_{h}| = O\left(n^{1+\frac{2}{\log n}}\right) = O(n)$, and by Theorem~\ref{th:smallcycles} we have
$$\left|\bigcup_{e \in E_h} C_e \right| = O\left(n \cdot \sqrt{|E_h|}\right) = O\left(n^{\frac{3}{2}}\right)$$

\end{description}

\noindent
The total number of edges in $R$ depends on values of $\sigma$ and $W$:
\begin{itemize}
\item for $3 \leq \sigma < 5$\\
the number of edges is $O\round{W \cdot n^{\frac 3 2}}$;

\item for $5 \leq \sigma < \log n$ and $W = \Omega\round{\round{n^{\frac{1}{2} - \frac{1}{\floor{(\sigma + 1)/2}}}} / \log n}$\\
the number of edges is $O\round{W \cdot \log n \cdot n^{1 + \frac{1}{\floor{(\sigma + 1)/2}}}}$;

\item for $\sigma \geq \log n$, or $\sigma \geq 5$ and $W = O\round{\round{n^{\frac{1}{2} - \frac{1}{\floor{(\sigma + 1)/2}}}} / \log n}$\\
the number of edges is $O\round{n^{\frac 3 2}}$.

\end{itemize}
\end{proof}

\noindent
Note that, in the case of unweighted graphs, the number of edges in $R$ is always $O\round{n^{3 / 2}}$. Thanks to Corollary~\ref{coro:alfabeta}, Algorithm \ResilientSpanner\ can also be used to compute a $\sigma$-resilient $(\alpha,\beta)$-spanner $R$ of an unweighted graph, for any $\sigma \geq \alpha+\beta$, containing $O\left(n^{3/2}\right)$ edges in the worst case.

\commento{
Since\nota{questa frase si riferisce all'algo, spostare nella relativa subsection} the time required to compute all those underlying spanners is $O(mn)$, 
in all those cases the time required to build a $\sigma$-resilient spanner is $O(mn)$.
Theorem~\ref{th:shvalpreserving} can also be applied to build $\sigma$-resilient $f$-spanners, where $f$ is a general \emph{distortion} function
as defined in~\cite{PettieSchema}, provided that $\sigma \geq f(1)$. Furthermore, 
if we wish to compute a $\sigma$-resilient $t$-spanner with $\sigma < t$, the same algorithm can still be applied starting from a $\sigma$-spanner instead of an $t$-spanner, yielding the same bounds given in Theorem~\ref{th:shvalpreserving}.
}

\commento{
Our results can be extended to weighted graphs\nota{riportare la prova per i weighted},
since Theorems~\ref{th:fewhighshval} and \ref{th:smallcycles} also hold for graphs with positive edge weights.
Let $w_{max}$ and $w_{min}$ be respectively the maximum and minimum edge weight in the graph, and  let $W = \frac{w_{max}}{w_{min}}$.
For either $\sigma > \log n$ or $\sigma \geq 5$ and $W = O\round{\round{n^{\frac{1}{2} - \frac{1}{\floor{(\sigma + 1)/2}}}} / \log n}$,  we can compute
a $\sigma$-resilient $t$-spanner
with $O(n^{\frac{3}{2}})$ edges.
In the remaining cases (either  $\sigma \geq 5$ and larger $W$ or $\sigma = 3, 4$), the
total number of edges becomes $O(W \cdot n^{\frac{3}{2}})$ \nota{PGF: perch\'e questo non si accorda con la nuova prova?}. 
}

\section{Conclusions and further work}\label{se_conclusions}

In this paper, we have investigated a new notion of resilience in graph spanners by introducing the concept of $\sigma$-resilient spanners. In particular, we have shown that it is possible to compute small stretch $\sigma$-resilient spanners of optimal size for graphs with small positive edge weights.

The techniques introduced for small stretch $\sigma$-resilient $t$-spanners can be used to turn any generic spanner (e.g., a fault-tolerant spanner, or, in the unweighted case, an $(\alpha,\beta)$-spanner, for $\sigma\geq \alpha+\beta > 3$) into a $\sigma$-resilient spanner, by adding a suitably chosen set of at most $O(W \cdot n^{3/2})$ edges (that is, $O(n^{3/2})$ in the unweighted case). 

We expect that in practice  
our $\sigma$-resilient $t$-spanners, for $\sigma \geq t >3$, will be substantially sparser than what it is implied by the bounds given in Theorem~\ref{th:shvalpreserving}, and thus of higher value in applicative scenarios. Towards this aim, we plan to perform a thorough experimental study.
Another intriguing question is whether our theoretical analysis on the number of edges that need to be added to a $t$-spanner in order to make it $\sigma$-resilient provides tight bounds, or whether it can be further improved for $\sigma > 3$. 

\bibliographystyle{plain}
\bibliography{vital}

\end{document}